\documentclass[11pt,a4paper]{article}
\usepackage{amsmath,amssymb,amsthm}
\usepackage[pdftex]{graphicx}
\usepackage[width=0.9\textwidth,indention=0mm,justification=centerlast,
 font=small,labelfont=bf,labelsep=period]{caption}
\usepackage[pdftoolbar=true,pdfpagemode=UseOutlines,
 bookmarks=true,bookmarksopen=true,bookmarksnumbered=true,
 pdffitwindow=true,pdfcenterwindow=true,pdfpagelayout=SinglePage,
 pdfhighlight=/P,colorlinks,linkcolor=blue,citecolor=red,urlcolor=blue,
 unicode=true]{hyperref}

\def\iu{{\rm i}}
\theoremstyle{plain}
\newtheorem{theorem}{Theorem}

\title{The tangential map and associated integrable equations}
\author{V.E. Adler\thanks{L.D. Landau Institute for Theoretical Physics, 1a Ak. Semenov pr.,
Chernogolovka 142432, Russia. E-mail: adler@itp.ac.ru}}
\date{8 June 2009}

\begin{document} \maketitle
\begin{abstract}
The tangential map is a map on the set of smooth planar curves. It satisfies the
3D-consistency property and is closely related to some well-known integrable
equations.
\end{abstract}

%-------------------------------------------------------------------------------
\section{Introduction}

Let smooth planar curves $C$, $C_1$ and $C_2$ be given. Draw the tangent line to $C$
through any point $r$ on this curve, and let it meet $C_1$ in a point $r_1$ and $C_2$
in a point $r_2$. Let the tangent lines to the respective curves through these points
meet in a point $r_{12}$. When the point $r$ moves along $C$, the point $r_{12}$ draws
a new curve $C_{12}$. Thus, a local mapping on the set of planar curves is defined,
\[
 F:(C,C_1,C_2)\mapsto C_{12},
\]
which will be referred to as the {\em tangential map}. The word ``local'' means that,
first, the mapping is defined not for all triples of curves since the tangent to the
curve $C$ may not intersect $C_1$ or $C_2$, and therefore only such curves or parts of
the curves are considered where the construction is possible; second, the mapping may
be multivalued since there may be several intersections, in this case a fixed branch
of the mapping is considered.

Some properties of the tangential map are studied in this paper. It turns out to be
rather simply related to the factorization of differential operators. In its turn,
this allows to establish the relation with such integrable equations as semidiscrete
($\Delta\Delta D$) Toda lattice and, under a reduction, Hirota equation
($\Delta\Delta$). One of modifications of discrete KP equation ($\Delta\Delta\Delta$)
appears in the discrete version of the tangential map. Thus, the tangential mapping is
not a quite new object, rather it is of certain interest as one more geometric
interpretation of well-known integrable equations.

%-------------------------------------------------------------------------------
\section{3D-consistency}

The main property of the tangential map is 3D-consistency. This means that if one
starts from the curves $C,C_1,C_2,C_3$ and constructs the curves $C_{ij}=F(C,C_i,C_j)$
then the curve $C_{123}$ constructed from the triple $C_i,C_{ij},C_{ik}$ is one and
the same for any permutation of $i,j,k$. Alternatively, this can be formulated as
follows.

\begin{theorem}\label{th:3Dcons}
The tangential map satisfies the (local) identity
\begin{equation}\label{3Dcons}
\begin{aligned}
C_{123}&= F(C_1,F(C,C_1,C_2),F(C,C_1,C_3))\\
 &= F(C_2,F(C,C_1,C_2),F(C,C_2,C_3))\\
 &= F(C_3,F(C,C_1,C_3),F(C,C_2,C_3)).
\end{aligned}
\end{equation}
\end{theorem}

The proof is given in the next Section. The combinatorial structure of identity
(\ref{3Dcons}) is represented by assigning the arguments of the mapping (the curves in
our case) to the vertices of a cube, and the mapping itself to the faces, as shown on
\hyperref[fig:3D]{Fig. \ref*{fig:3D}}. The $N$-fold iteration of the mapping is
associated with an $(N+1)$-dimensional cube.

\begin{figure}[b!]
\centerline{\includegraphics[width=5cm]{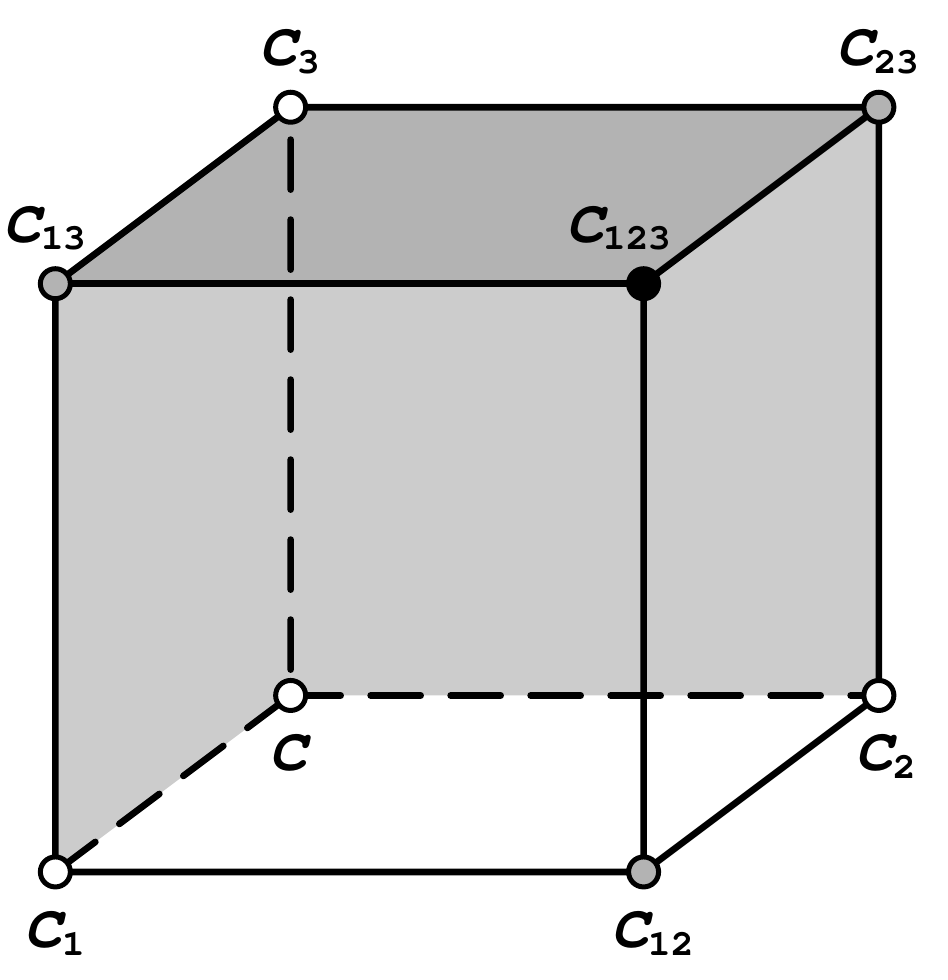}}
\caption{3D-consistency, or consistency around a cube. The white vertices correspond
to the given curves, the shaded faces show one of three possible construction ways of
the curve corresponding to the black vertex.}\label{fig:3D}
\end{figure}

The notion of 3D-consistency was formulated in \cite{NW,BS} in connection to the
discrete integrable equations of the difference KdV type (the fields in the vertices
of the cube) or to the Yang-Baxter type mappings (the fields on the edges of the
cube). Both types of equations appear as nonlinear superposition principle for
Darboux-B\"acklund transformations and are two-dimensional: two discrete independent
variables correspond to the shifts along the edges of an elementary square. In
contrast, the case of the tangential map is related to a three-dimensional equation:
in addition to the discrete variables a continuous one appears corresponding to a
parameter along the curves. Another important distinction is the asymmetry of the
tangential map: the roles of the involved curves are obviously different. In
particular, the formulas from the next section make clear that the construction of
$C_{ij}$ from $C$, $C_i$, $C_j$ is described by a differential rational mapping, while
the construction, for instance, of $C_j$ from $C$, $C_i$, $C_{ij}$ requires an
additional integration. This explains the choice of the set $C$, $C_i$, $C_j$,
$C_k,\dots$ as preferable initial data, rather that the sequence $C$, $C_i$, $C_{ij}$,
$C_{ijk},\dots$, as it is usual in the standard formulation of Yang-Baxter mappings
\cite{V}.

\begin{figure}[t!]
\centerline{\includegraphics[width=8cm]{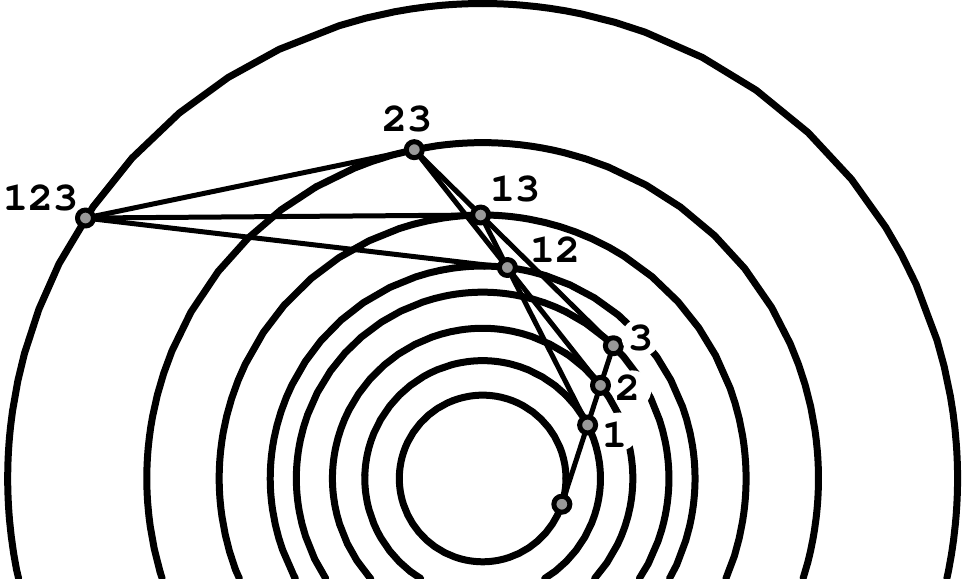}}
\caption{3D-consistency of the tangential map in the simplest case of concentric circles}
\label{fig:circles}
\end{figure}

The simplest example illustrating the tangential map and its 3D-consist\-ency property
is given by the family of concentric circles (see \hyperref[fig:circles]{Fig.
\ref*{fig:circles}}), or, in slightly more general form, by the family of logarithmic
spirals defined by equations $\rho_i=c_ie^{\gamma\varphi}$ in polar coordinates. The
fact that the tangential map does not lead out of these families is clear from the
invariance of the construction with respect to the rotations and scalings. However,
the concurrence of the last three tangents is not spontaneously obvious (as an
additional feature, in the case of family of circles, the points
$r_{12},r_{13},r_{23}$ lie on the circle with $Or_{123}$ as diameter, where $O$ is the
center of the family). Of course, this can be proved by elementary methods, but the
point is that the concurrence of the tangents occurs in much more general situation.
In \hyperref[s:log-spiral]{Section \ref*{s:log-spiral}} it will be demonstrated that
this example corresponds to the simplest case of the factorization of differential
operators with constant coefficients. \hyperref[fig:spiral]{Fig.~\ref*{fig:spiral}}
illustrates a nice property of the logarithmic spiral (in contrast to the previous
plot, this one contains only three generations of points $r,r_i,r_{ij}$, that is the
triple intersections of the lines are not shown). This property can be formulated as
the identity $F(C,C,C)=C$ for any branch of the map $F$.

\begin{theorem}
Consider the intersection points of the logarithmic spiral with its tangent. The
tangents through these points meet on the same spiral.
\end{theorem}

\begin{figure}[t!]
\centerline{\includegraphics[width=8cm]{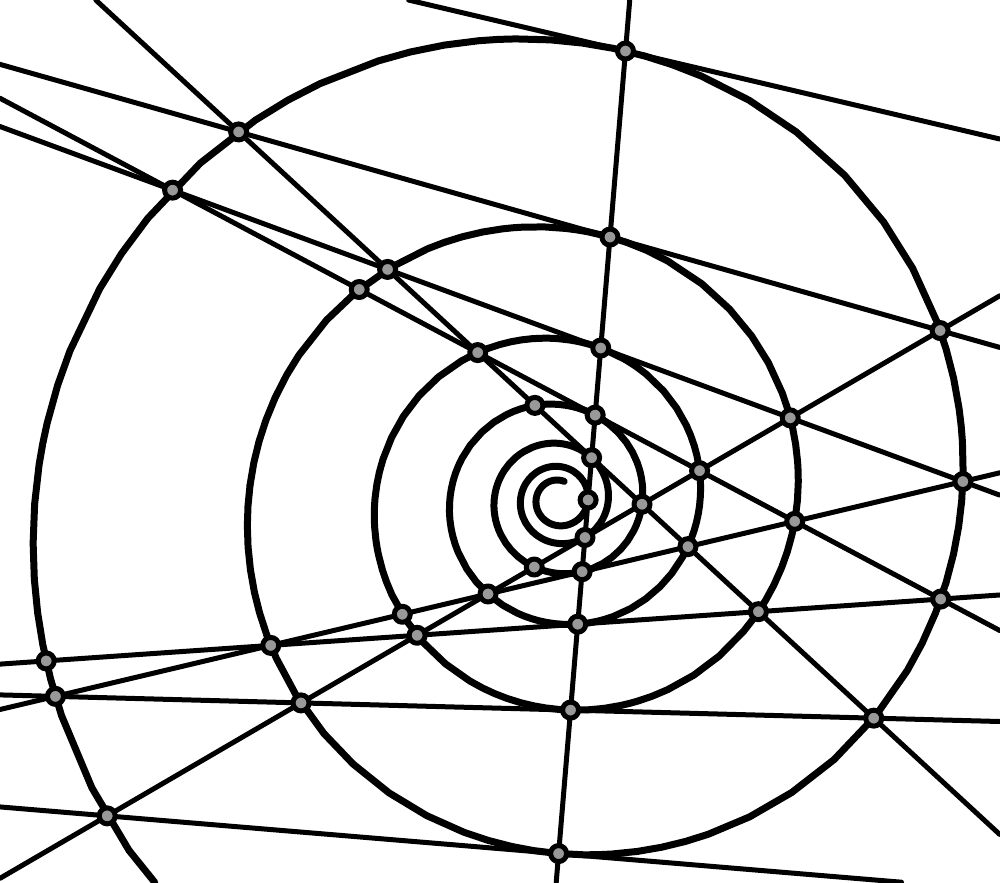}}
\caption{Jacob Bernoulli said ``eadem mutata resurgo'' on the logarithmic spiral;
now in infinite number of ways}
\label{fig:spiral}
\end{figure}

%-------------------------------------------------------------------------------
\section{Factorization of differential operators}

Let the curve $C$ be given in parametric form $r=r(t)$. Then the point of intersection
with the curve $C_i$ is given, in the affine coordinates, by an equation of the form
\begin{equation}\label{ri}
 r_i(t)=r(t)+a^i(t)\dot r(t),\qquad \dot r:=D(r):=\frac{dr}{dt}
\end{equation}
with some coefficient $a^i$ (here and further on the subscript $i$ marks the
quantities associated with the edge $CC_i$ of the combinatorial cube). This equation
defines a parametrization of the curve $C_i$. It plays the role of an auxiliary linear
problem for one branch of the tangential map. The curve $C_{ij}$ is defined from the
compatibility condition
\begin{align*}
  r_{ij}&=r_j+a^i_j\dot r_j=r+a^j\dot r+a^i_jD(r+a^j\dot r)\\
        &=r_i+a^j_i\dot r_i=r+a^i\dot r+a^j_iD(r+a^i\dot r),
\end{align*}
where the coefficient $a^i_j$ correspond to the edge $C_jC_{ij}$.

Equating the coefficients for the linearly independent vectors $\dot r,\ddot r$ yields
the equations
\begin{equation}\label{aa}
 a^ja^i_j=a^ia^j_i,\quad (1+\dot a^j)a^i_j+a^j=(1+\dot a^i)a^j_i+a^i
\end{equation}
which can be solved as the differential mapping $(a^i,a^j)\mapsto(a^i_j,a^j_i)$
\begin{equation}\label{Ta}
 a^i_j=\frac{(a^i-a^j)a^i}{a^i-a^j+a^i\dot a^j-\dot a^ia^j}.
\end{equation}
This is the formula which defines the action of the tangential map on the coefficients
$a$. Alternatively, one can use the first of equations (\ref{aa}) in order to
introduce the potential $v$ accordingly to the formula $a^i=v/v_i$, then the second
equation rewrites as the differential mapping $f:(v,v_i,v_j)\mapsto v_{ij}$,
\begin{equation}\label{TTv}
 v_{ij}=\frac{v_iv_j}{v}+\frac{\dot v_iv_j-v_i\dot v_j}{v_j-v_i}.
\end{equation}
The property of 3D-consistency is formulated in terms of the variables $a$ as the
commutativity of the operators $T_i:a^j\to a^j_i$ which define the shift along the
edges $CC_i$:
\begin{equation}\label{TTa}
 T_iT_j(a^k)=T_jT_i(a^k),
\end{equation}
and in terms of the variables $v$ as the identity of type (\ref{3Dcons}):
\begin{equation}\label{TTTv}
\begin{aligned}
v_{123}&= f(v_1,f(v,v_1,v_2),f(v,v_1,v_3))\\
 &= f(v_2,f(v,v_1,v_2),f(v,v_2,v_3))\\
 &= f(v_3,f(v,v_1,v_3),f(v,v_2,v_3)).
\end{aligned}
\end{equation}
Both identities can be proved straightforwardly, although the computation is rather
tedious. It can be avoided by the following argument.

\begin{proof}[Proof of {\hyperref[th:3Dcons]{Theorem \ref*{th:3Dcons}}}]
\label{proof:3Dcons}
The above compatibility condition is equivalent to the equality
\begin{equation}\label{fac2}
 (1+a^j_iD)(1+a^iD)=(1+a^i_jD)(1+a^jD),
\end{equation}
that is the definition of the tangential map amounts to the reconstruction of an
ordinary second order differential operator from its kernel, under the condition of
unitary constant term which is equivalent to affine normalization. Consider the
differential operator
\[
 A=(1+T_i(a^k_j)D)(1+a^j_iD)(1+a^iD)
\]
corresponding to one of three possible ways of computing $r_{ijk}$. Accordingly to
(\ref{fac2}), $A$ is divisible from the right not only by $1+a^iD$, but also by
$1+a^jD$. Moreover, two left factors of $A$ can be rewritten, again accordingly to
(\ref{fac2}), as $(1+T_i(a^j_k)D)(1+a^k_iD)$, that is operator $A$ does not
changes under permutation of $j$ and $k$. But this means that it is divisible from the
right by $1+a^kD$ as well. Therefore, the kernel of $A$ is invariant with respect to
any permutation of indices. Since a differential operator is uniquely defined by its
kernel (up to a scalar factor which is fixed here by the condition that the constant
term is unitary), hence the operator $A$ itself is invariant with respect to the
permutations.
\end{proof}

Now it is clear that an $N$-fold tangential map corresponds to a differential operator
of order $N$ divisible from the right by operators $1+a^iD$, $i=1,\dots,N$. This
immediately leads to the Wronskian formula (for each of two components of $r$)
\[
r_{1,2,\dots,N}=\frac{\det
 \begin{pmatrix}
  r      & \varphi_1      & \varphi_2      & \dots & \varphi_N     \\
  \dot r & \dot\varphi_1  & \dot\varphi_2  & \dots & \dot\varphi_N \\
  \vdots & \vdots & \vdots & & \vdots \\
  D^N(r) & D^N(\varphi_1) & D^N(\varphi_2) & \dots & D^N(\varphi_N)
 \end{pmatrix}}
 {\det\begin{pmatrix}
  \dot\varphi_1 & \dot\varphi_2 & \dots & \dot\varphi_N \\
  \vdots & \vdots & & \vdots \\
  D^N(\varphi_1) & D^N(\varphi_2) & \dots & D^N(\varphi_N)
 \end{pmatrix}}
\]
where $a^i=-\varphi_i/\dot\varphi_i$.

More simple mappings of type (\ref{Ta}) and (\ref{TTv}) are obtained from the
factorization of operators normalized by the condition of unitary leading term:
\[
 (D-a^j_i)(D-a_i)=(D-a^i_j)(D-a_j),
\]
which is equivalent to
\[
 (T_i-1)(a^j)=(T_j-1)(a^i),\quad \dot a^i-\dot a^j=a_ia^j_i-a^ja^i_j,
\]
and brings to the maps
\begin{equation}\label{Ta'}
 a^i_j=a^i+\frac{\dot a^i-\dot a^j}{a^i-a^j}
\end{equation}
and (under the substitution $a^i=v_i-v$)
\begin{equation}\label{TTv'}
 v_{ij}=v_i+v_j-v+\frac{\dot v_i-\dot v_j}{v_i-v_j}
\end{equation}
The 3D-consistency of these maps is proved analogously. Clearly, equations (\ref{Ta}),
(\ref{TTv}) and (\ref{Ta'}), (\ref{TTv'}) are interpreted as 3-dimensional equations
on $\mathbb Z^2\times\mathbb R$, with the fields $a$ corresponding to the edges of the
lattice and $v$ corresponding to the vertices. These equations are related via simple
substitutions to the semidiscrete Toda lattice, introduced in \cite{LPS} for the first
time (to the best of author's knowledge), see also \cite{AS}.

%-------------------------------------------------------------------------------
\section{Examples and reductions}

\subsection{Logarithmic spirals}\label{s:log-spiral}

It is convenient to use the complex notation in this example, assuming
$r=e^{(\gamma+\iu)t}$ (the case $\gamma=0$ correspond to the circle). Then
$r_k=r+a^k\dot r=(1+\gamma a^k+\iu a^k)e^{(\gamma+\iu)t}$ and these curves are
homothetic to the original one if and only if the coefficients $a^k$ are constant. The
action of the map (\ref{Ta}) on the constant coefficients is identical: $a^k_j=a^k$,
therefore the tangential map amounts to the rotational dilation
\[
 r_{jk}=(1+\gamma a^j+\iu a^j)(1+\gamma a^k+\iu a^k)r
\]
which preserves the family of curves under consideration. The $N$-fold mapping is
given by analogous explicit formula, so that this example can be considered trivial.
However, even this example demonstrates that the established relation between the
tangential map and differential operators is not one-to-one and depends on the choice
of initial curve and its parametrization. The mappings corresponding to the same
operators (that is, with the same coefficients $a$) can be regarded as locally
equivalent, but the global picture may be quite different. For example, the tangential
map has four branches in the case of concentric circles (real if the radius of $C$ is
less than the radii of $C_1,C_2$) and infinite number of branches in the case of
logarithmic spirals.

In order to obtain the auto-mapping shown on \hyperref[fig:spiral]{Fig.
\ref*{fig:spiral}}, one has to assume the additional constraint $r_k(t)=r(t+\delta_k)$
(then $r_{jk}(t)=r(t+\delta_j+\delta_k)$, $r_{jkl}(t)=r(t+\delta_j+\delta_k+\delta_l)$
and so on), that is
\[
 e^{(\gamma+\iu)(t+\delta_k)}=(1+\gamma a^k+\iu a^k)e^{(\gamma+\iu)t}.
\]
This implies that $\delta_k$ are roots of the transcendental equation
\[
 \cos\delta-\gamma\sin\delta=\exp(-\gamma\delta),
\]
and the coefficients $a^k$ are expressed by the formula
\[
 a^k=\exp(\gamma\delta_k)\sin\delta_k.
\]
It can be derived from here (this is left to the reader as an easy exercise), that the
boundary of the domain free of the lines on Fig. \ref{fig:spiral} is approximated by a
parabola. The numeric values for this plot are $\gamma=0.1$, $\delta_1=5.24$,
$\delta_2=7.25,\dots$

%-------------------------------------------------------------------------------
\subsection{Any curve}

The previous example suggests that a picture with good global behavior of the curves
can be obtained if the starting curve is the circle $r=e^{\iu t}$ again, and the
coefficients $a^k(t)$ are almost constant functions with periods commensurable with
$\pi$. For instance, the left plot on \hyperref[fig:ex2]{Fig. \ref*{fig:ex2}}
corresponds to
\[
 a^1=1+\frac15\cos\frac32t,\quad
 a^2=2+\frac1{10}\cos\Bigl(\frac{t}2+\frac{\pi}4\Bigr).
\]
The right plot corresponds to the choice
\[
 a^1=4+\sin t,\quad a^2=4+\cos t,
\]
here the tangential map brings to the curve with cusps.

\begin{figure}
\centerline{\includegraphics[height=6cm]{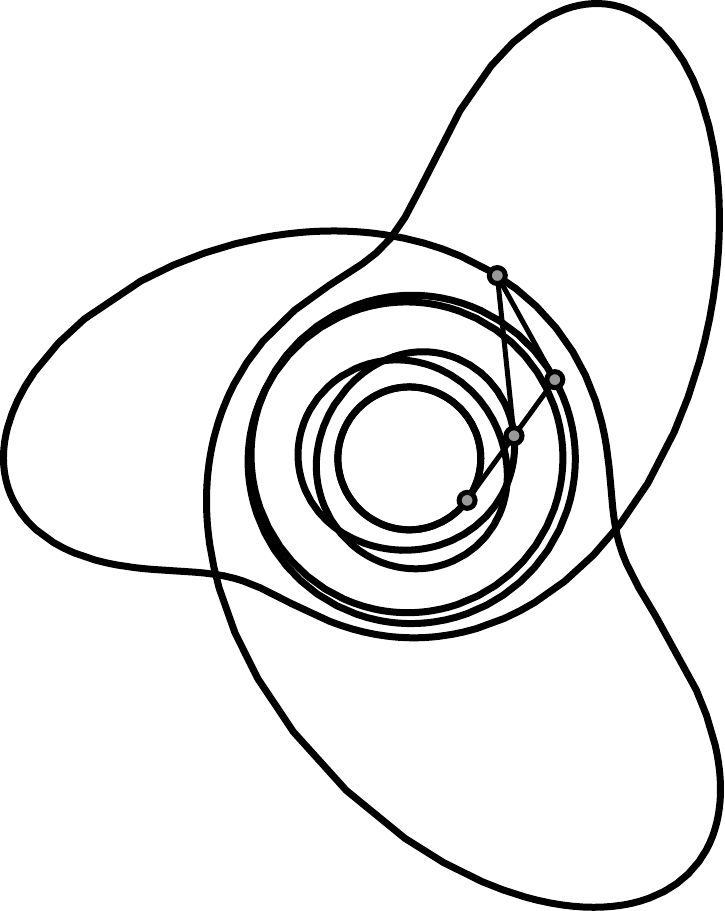}\qquad\qquad
\includegraphics[height=6cm]{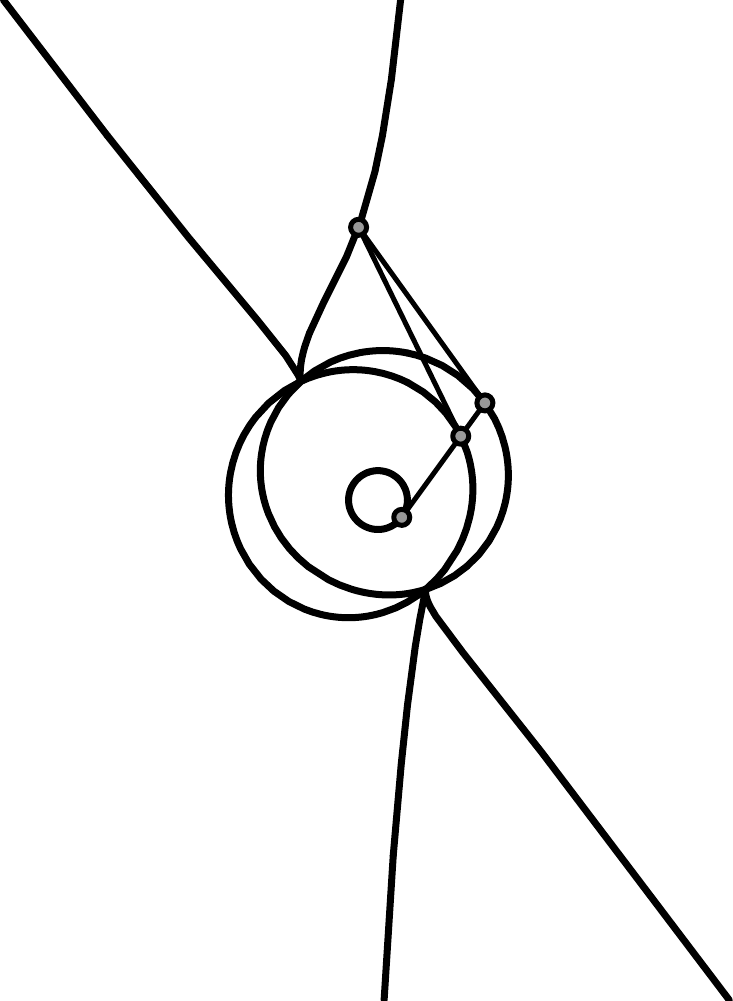}}
\caption{Examples of the tangential map}\label{fig:ex2}
\end{figure}

%-------------------------------------------------------------------------------
\subsection{Loxodromes}

We will say, slightly abusing the terminology, that a curve $\widetilde C$ is a {\em
loxodrome} for a given curve $C$ if it intersects the tangents to $C$ under a constant
angle $\gamma$ (in particular, if $\gamma=\pi/2$, then $\widetilde C$ is an involute
of $C$).

\hyperref[th:loxodrome]{Theorem \ref*{th:loxodrome}} below demonstrates that the
tangential map preserves this type of relation between the curves (see
\hyperref[fig:loxodrome]{Fig. \ref*{fig:loxodrome}}). As a preliminary, it is
convenient to introduce the parameter on the curve $C$ accordingly to equations
\begin{equation}\label{y}
 \dot r=y(t)\tau,\quad \dot\tau=\nu,\quad \dot\nu=-\tau
\end{equation}
where $\tau,\nu$ are unit tangent and normal vectors. Obviously, the function $y$ is
the radius of curvature and the relation to the natural parametrization by arc length
$r=r(s)$ is
\[
 y(t(s))=1/\kappa(s),\quad dt=\kappa(s)ds.
\]
This choice of parametrization is explained by the fact that its form is preserved for
the curve $\widetilde C$ as well. Indeed, the equation of a point on this curve is
$\tilde r= r+a\dot r=r+ay\tau$. Then
\[
 \dot{\tilde r}=(y+D(ay))\tau+ay\nu
\]
and since $\tilde r$ meets $\tau$ under the constant angle $\gamma$ (reckoned in the
direction of the normal $\nu$, for definiteness), hence
\begin{equation}\label{ya}
 y+D(ay)=ay\cot\gamma.
\end{equation}
In virtue of this constraint, the equalities
\[
 \dot{\tilde r}=\tilde y\tilde\tau,\quad
 \dot{\tilde\tau}=\tilde\nu,\quad \dot{\tilde\nu}=-\tilde\tau
\]
hold, with
\[
 \tilde y=\frac{ay}{\sin\gamma},\quad
 \tilde\tau=\tau\cos\gamma+\nu\sin\gamma,\quad
 \tilde\nu=\nu\cos\gamma-\tau\sin\gamma.
\]
Thus, we have proved that the choice of $t$ as the parameter brings to equations of
the type (\ref{y}) for the loxodrome $\widetilde C$ as well. If the function $y$ is
given then the constraint (\ref{ya}) is the determining equation for the coefficient
$a$, and the loxodrome is constructed by any solution of it.

\begin{figure}[t!]
\centerline{\includegraphics[width=7cm]{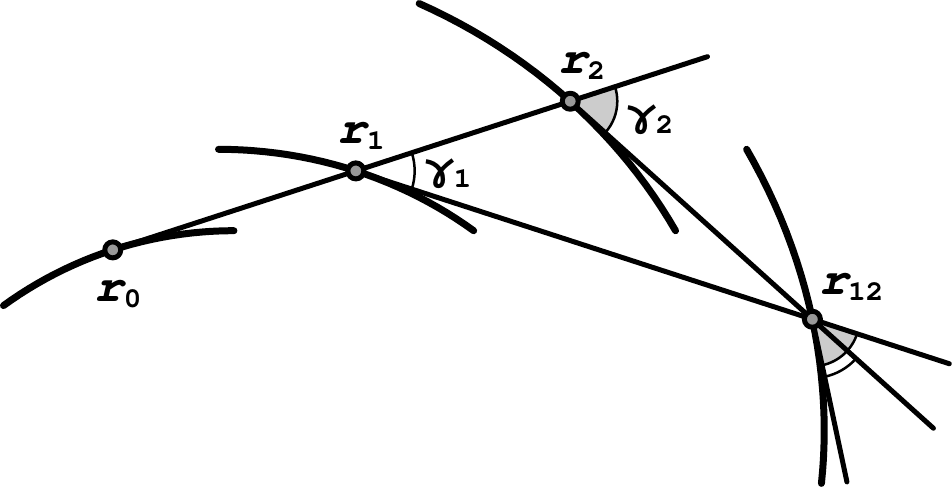}}
\caption{Loxodromic reduction of the tangential map}\label{fig:loxodrome}
\end{figure}

\begin{theorem}\label{th:loxodrome}
Let curves $C_i$ and $C_j$ meet tangents to a curve $C$ under constant angles
$\gamma^i$ and $\gamma^j$ respectively, and $\gamma^i\ne\gamma^j$. Then the curve
$C_{ij}=F(C,C_i,C_j)$ meets tangents to $C_i$ under the angle $\gamma^j$ and tangents
to $C_j$ under the angle $\gamma^i$.
\end{theorem}
\begin{proof}
Assume the parametrization (\ref{y}) for the curve $C$. Then, as it was shown before,
the parameters $a^k$ of the tangential map and the functions $y_k$ for the loxodromes
are related by
\[
 y+D(a^ky)=a^ky\cot\gamma^k,\quad y_k=\frac{a^ky}{\sin\gamma^k},\quad k=i,j.
\]
It is easy to prove that in virtue of these constraints the tangential map (\ref{Ta})
takes the form
\[
 a^i_j=\frac{a^i/a^j-1}{\cot\gamma^j-\cot\gamma^i}
\]
and moreover, the identity
\[
 y_j+D(a^i_jy_j)=a^i_jy_j\cot\gamma^i
\]
holds, which is the constraint of the form (\ref{ya}) for the functions $y_j$ and
$a^i_j$. But this means exactly that the curve $C_{ij}$ is a loxodrome for $C_j$,
corresponding to the angle $\gamma^i$.
\end{proof}

Formally, this statement looks like well-known Bianchi theorem on the permutability of
Darboux-B\"acklund type transformations (see \hyperref[fig:Bianchi]{Fig.
\ref*{fig:Bianchi}}). However, it is clear from the proof that in our case the
situation is more simple: indeed, B\"acklund transformations amount to solving of
Riccati equations, while the construction of the loxodromes amounts, accordingly to
(\ref{ya}), to a simple quadrature. The superposition principle for the
transformations under consideration turns out to be linear:
\[
 \sin(\gamma^i-\gamma^j)y_{ij}=\sin(\gamma^i)y_i-\sin(\gamma^j)y_j.
\]
A genuine Darboux transformation leading to the nonlinear superposition principle is
provided by the reduction presented in the next example.

\begin{figure}[t!]
\centerline{\includegraphics[width=3.5cm]{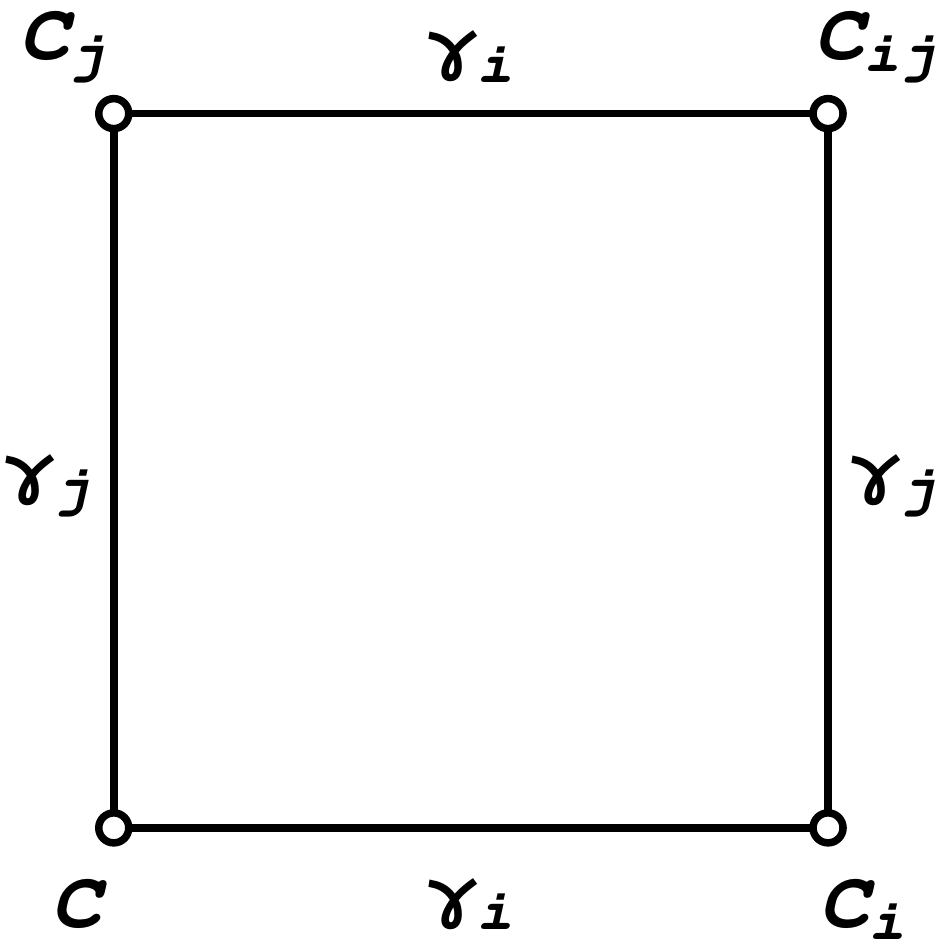}}
\caption{Bianchi diagram}\label{fig:Bianchi}
\end{figure}

%-------------------------------------------------------------------------------
\subsection{Darboux transformation}

Let $H_{\gamma}$ be the homothety with a coefficient $\gamma\ne1$ with respect to a
fixed point (origin) in the plane. We will say that the curves $C$, $\widetilde C$ are
in the {\em tangential correspondence} with parameter $\gamma$ if the tangent to the
curve $C$ through any point $r$ meets $H_{\gamma}(\widetilde C)$ in the point
$H_{\gamma}(\tilde r)$, and the tangent to $\widetilde C$ through $\tilde r$ meets
$H_{\gamma}(C)$ in the point $H_{\gamma}(r)$ (see Fig. \ref{fig:t-corr}).

\begin{figure}[b!]
\centerline{\includegraphics[width=7cm]{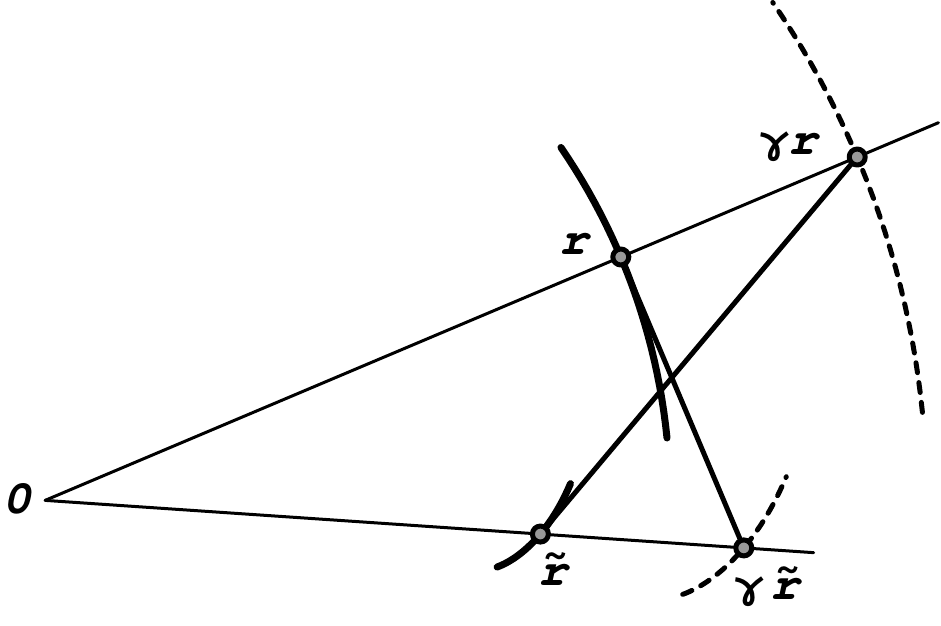}}
\caption{Tangential correspondence}\label{fig:t-corr}
\end{figure}

As in the example with the loxodromes, this notion defines some reduction of the
tangential map which can be studied more conveniently in some special parametrization
of the curves. The points on the curves $C$, $\widetilde C$ are related by equations
of the form
\[
 \gamma\tilde r(t)=r(t)+a(t)\dot r(t),\quad \gamma r(t)=\tilde r(t)+b(t)\dot{\tilde r}(t),
\]
which imply that the vector-functions $r(t),\tilde r(t)$ satisfy linear second order
ODEs
\begin{gather*}
 ab\ddot r+(\dot ab+a+b)\dot r+(1-\gamma^2)r=0,\\
 ab\ddot{\tilde r}+(a\dot b+a+b)\dot{\tilde r}+(1-\gamma^2)\tilde r=0.
\end{gather*}
Therefore, the ratio of the first and the last coefficients is the same for $r$ and
$\tilde r$, that is it is an invariant of the tangential correspondence. It is
convenient to use such a parametrization that this ratio is constant. Let
\[
 \lambda ab=1-\gamma^2,\quad \lambda=(1-\gamma^2)\mu,
\]
then the equations take the form
\begin{equation}\label{rurr}
 \ddot r+u\dot r+\lambda r=0,\quad
 \ddot{\tilde r}+\tilde u\dot{\tilde r}+\lambda\tilde r=0
\end{equation}
where functions $u,\tilde u$ are related to the coefficient $a$ via the pair of
Riccati equations
\[
 \dot a+1-ua+\mu a^2=0,\quad -\dot a+1-\tilde ua+\mu a^2=0.
\]
The first one is linearized by the substitution $a=-\phi/\dot\phi$ which brings to
equation $\ddot\phi+u\dot\phi+\mu\phi=0$. Thus, the curve $\widetilde C$ is
constructed by use of a particular scalar solution of the differential equation for
the original curve $C$, at the value of the parameter $\lambda=\mu$, that is the
tangential correspondence is nothing but an example of Darboux transformation.

\begin{theorem}
Let a curve $C$ be in the tangential correspondence with curves $C_i$, $C_j$, with
parameters $\gamma^i$, $\gamma^j$ respectively, and $\gamma^i\ne\gamma^j$. Then an
unique curve $C_{ij}$ exists which is in the tangential correspondence with the curves
$C_i$, $C_j$, with the parameters $\gamma^j$, $\gamma^i$ respectively. Moreover,
\[
 H_{\gamma^i\gamma^j}(C_{ij})=F(C,H_{\gamma^i}(C_i),H_{\gamma^j}(C_j)).
\]
\end{theorem}
\begin{proof}
It follows directly from the definitions of the tangential map and the tangential
correspondence that if such a curve $C_{ij}$ exists then it is unique and is given by
the above formula. So we only need to verify that this curve is indeed in the
tangential correspondence with $C_i$, $C_j$, by use of the relations
\[
 \dot a^k+1=ua^k-\mu^k(a^k)^2,\quad u_k=-u+2\mu^ka^k+\frac{2}{a^k},\quad k=i,j.
\]
It is easy to prove that the formula (\ref{Ta}) takes, in virtue of these constraints,
the algebraic form
\begin{equation}\label{TaF4}
 a^i_j=\frac{a^i-a^j}{a^j(\mu^ia^i-\mu^ja^j)},
\end{equation}
and that the coefficient $A=a^i_j$ identically satisfies the Riccati equation
\[
 \dot A+1-u_jA+\mu^iA^2=0.
\]
This means that the curves $C_{ij}$, $C_j$ are in tangential correspondence, with
parameter $\gamma^i$.
\end{proof}

It should be noticed that equation (\ref{rurr}) for $r$ defines the spectral problem
for the $\sinh$-Gordon equation, mapping (\ref{TaF4}) is equivalent to the mapping
($F_4$) from \cite{ABS4}, and the substitution $a^i=v/v_i$, $u=2\dot v/v$ brings to
the equation
\[
 v_{ij}(v_j-v_i)=v(\mu_iv_j-\mu_jv_i)
\]
which is equivalent to the Hirota equation \cite{H}.

Analogously, the mapping (\ref{Ta'}) admits the similar reduction
\[
 \dot a^k=u+\gamma^k-(a^k)^2,\quad u_k=-u-2\gamma^k+2(a^k)^2
\]
which brings to the mapping
\[
 a^i_j=-a^j+\frac{\gamma^i-\gamma^j}{a^i-a^j}.
\]
This is one of the forms of nonlinear superposition principle of Darboux
transformation for the Shr\"odinger operator \cite{VS}.

%-------------------------------------------------------------------------------
\section{Further generalizations}

\subsection{Discrete tangential map}

Let us consider discrete curves $r=r(n)$, then an analog of equation (\ref{ri}) reads
\[
 r_i=r+a^i(T-1)(r),\quad T:n\mapsto n+1.
\]
The compatibility condition of such equations
\begin{align*}
 r_{ij}&=(1+a^i_j(T-1))(1+a^j(T-1))(r)\\
       &=(1+a^j_i(T-1))(1+a^i(T-1))(r).
\end{align*}
brings to the relations
\[
 T(a^j)a^i_j=T(a^i)a^j_i,\quad (1-a^j)a^i_j+a^j=(1-a^i)a^j_i+a^i,
\]
which yield a discrete analog of tangential map (\ref{Ta}):
\begin{equation}\label{dTa}
 T_j(a^i)=a^i_j=\frac{(a^i-a^j)T(a^i)}{(1-a^j)T(a^i)-(1-a^i)T(a^j)}.
\end{equation}
The substitution $a^i=T(v)/v_i$ brings to an analog of map (\ref{TTv}):
\begin{equation}\label{dTTv}
 f:(v,v_i,v_j)\mapsto v_{ij}=\frac{v_iv_jT(v_j-v_i)}{T(v)(v_j-v_i)}
  +\frac{T(v_i)v_j-v_iT(v_j)}{v_j-v_i}.
\end{equation}
The symmetric form of this equation
\[
 \frac{T(v_j-v_i)}{T(v)}+\frac{T(v_i)-v_{ij}}{v_i}+\frac{v_{ij}-T(v_j)}{v_j}=0,
\]
demonstrates that the shift $T$ is actually on the equal footing with $T_i$ and $T_j$.
Difference substitutions relate this equation to the discrete Toda and KP equations
(in particular, the variable $v$ is identified as the wave function of the linear
problem for KP equation \cite{KS}). Alternative geometric interpretations of these
equations can be found in the papers \cite{D,KS}, see also \cite{ABS9} where a general
theory of this class of equations is developed.

The 3D-consistency property of the maps (\ref{dTa}) and (\ref{dTTv}) is formulated by
the same general identities (\ref{TTa}), (\ref{TTTv}) as in the continuous case, and
is proved along the lines of the proof of \hyperref[th:3Dcons]{Theorem
\ref*{th:3Dcons}}. There exist also the simple geometric explanation of this
property\footnote{This proof is due to W.K. Schief}: the triangles
$r_{12}(n)r_{13}(n)r_{23}(n)$ and $r_{12}(n+1)r_{13}(n+1)r_{23}(n+1)$ are perspective
with respect to the line $r(n+1)r(n+2)$ (marked by $n+1$ on
\hyperref[fig:discrete]{Fig. \ref*{fig:discrete}}), therefore, accordingly to
Desargues theorem, the lines $r_{12}(n)r_{12}(n+1)$, $r_{13}(n)r_{13}(n+1)$ and
$r_{23}(n)r_{23}(n+1)$ are concurrent, as required.

\begin{figure}[t!]
\centerline{\includegraphics[width=8cm]{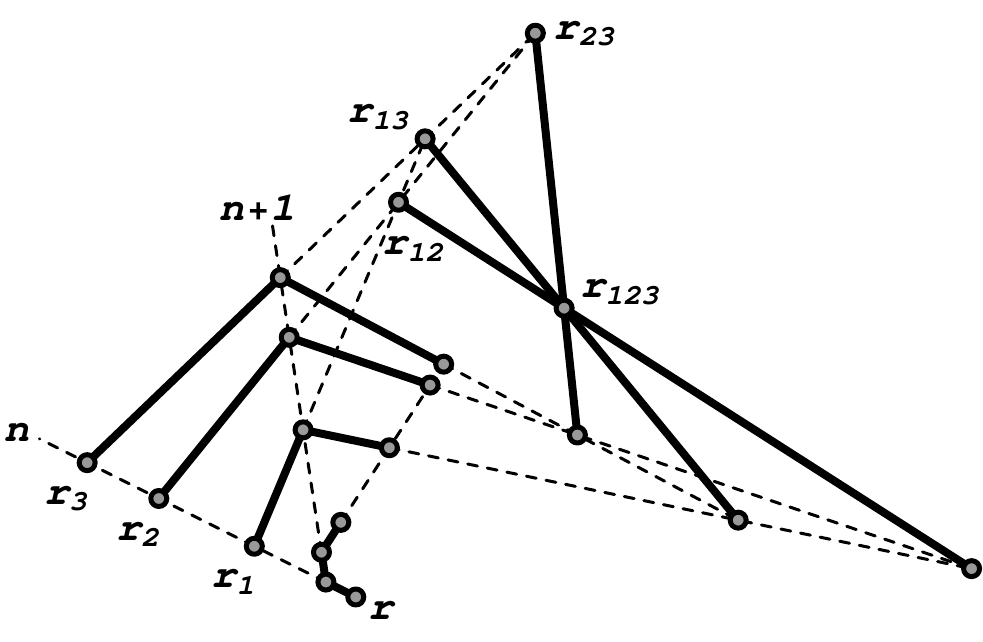}}
\caption{3D-consistency of the discrete tangential map}\label{fig:discrete}
\end{figure}

%-------------------------------------------------------------------------------
\subsection{The higher order maps}

Equation (\ref{ri}) can be replaced by
\[
 r_i(t)=r(t)+a^{1,i}(t)\dot r(t)+\dots+a^{m,i}(t)D^m(r(t)),
\]
for the curves in a space of dimension greater or equal to $2m$. In the discrete case
one has analogously
\[
 r_i(n)=r(n)+a^{1,i}(n)(r(n+1)-r(n))+\dots+a^{m,i}(n)(r(n+m)-r(n)).
\]
It is not difficult to demonstrate that the compatibility condition
$T_jT_i(r)=T_iT_j(r)$ is equivalent to a system of $2m$ equations which can be solved
in the form of differential or difference mapping
\[
 (a^{1,i},\dots,a^{m,i},a^{1,j},\dots,a^{m,j})\mapsto
 (a^{1,i}_j,\dots,a^{m,i}_j,a^{1,j}_i,\dots,a^{m,j}_i),
\]
with the derivatives or shifts up to $m$-th order. These mappings are 3D-consistent,
which can be proved by arguments analogous to the case $m=1$ just considered.
Unfortunately, the explicit form of these mappings is too bulky already at $m=2$, so
it would be desirable to find some reduction lowering their order and/or number of
fields.

\paragraph{Acknowledgements.} This work was supported by RFBR grants 08-01-00453 and
SS-3472.2008.2.

%-------------------------------------------------------------------------------

\end{document}